\documentclass[10pt, conference, letterpaper]{IEEEtran}
\IEEEoverridecommandlockouts

\usepackage{cite}
\usepackage{amsmath,amssymb,amsfonts}
\usepackage{algorithm}
\usepackage{algorithmic}
\usepackage{graphicx}
\usepackage{textcomp}
\usepackage{xcolor}
\usepackage{subcaption}
\usepackage{bm}

\usepackage{hyperref}
\hypersetup{
	colorlinks=true,
	linkcolor=blue,
	citecolor=blue,
	urlcolor=black
}

\newtheorem{theorem}{Theorem}
\newtheorem{proposition}{Proposition}

\newtheorem{lemma}{Lemma}
\newtheorem{assumption}{Assumption}
\renewcommand{\IEEEQED}{\IEEEQEDopen}
\makeatletter

\let\HyThm@begintheorem\@begintheorem
\def\@begintheorem#1#2{%
	\phantomsection\HyThm@begintheorem{#1}{#2}}
\let\HyThm@opargbegintheorem\@opargbegintheorem
\def\@opargbegintheorem#1#2#3{%
	\phantomsection\HyThm@opargbegintheorem{#1}{#2}{#3}}
\makeatother

\DeclareMathOperator{\bpi}{\bm{\pi}}
\DeclareMathOperator{\bomega}{\bm{\omega}}

\DeclareMathOperator{\Ind}{\mathbb{I}}
\DeclareMathOperator{\Ex}{\mathbb{E}}
\DeclareMathOperator{\Pb}{\mathbb{P}}
\newcommand{\wtilde}[1]{\widetilde{#1}}

\def\BibTeX{{\rm B\kern-.05em{\sc i\kern-.025em b}\kern-.08em
		T\kern-.1667em\lower.7ex\hbox{E}\kern-.125emX}}

\newif\ifisreport
\isreporttrue 

\begin{document}

\title{Load Balancing with Partial Queue Information: Threshold Optimality and Indexability}
\author{
	\IEEEauthorblockN{Sathwik Chadaga and Eytan Modiano}
	\IEEEauthorblockA{Laboratory for Information and Decision Systems, Massachusetts Institute of Technology, Cambridge, MA}
	\vspace{-16pt} \thanks{This work was supported by NSF grants CNS-2148183, CNS-2106268, and CNS-2148128.}
}

\maketitle

\begin{abstract}
We consider the problem of load balancing in a system with one dispatcher and $N$ parallel servers. The dispatcher must select one server to dispatch new jobs at every time-step and each server buffers incoming jobs in a queue. However, the dispatcher does not know the servers' backlogs and must make dispatching decisions based on previous observations. The dispatcher's objective is to dispatch jobs to the shortest queue. This problem can be formulated as a restless multi-arm bandit (RMAB) problem where each arm's state is its corresponding belief vector. Our goal is to verify Whittle indexability for this problem and derive a low complexity Whittle index policy. Previous Whittle indexability results cannot be directly applied due to the multi-dimensional nature of the belief vector. To overcome this issue, we define the RMAB state as the tuple of the most recent backlog observation and the time since this observation. We consider two model variations, a standard finite queue model and a blocking queue model. We show that the single-arm decoupled problems of both these models have threshold optimal solutions under some assumptions. For the standard finite queue model, we prove indexability and derive the Whittle index policy in closed form. For the blocking queue model, we derive a sufficient condition for indexability under threshold optimality and use it to show indexability for some special cases.
\end{abstract}


\section{Introduction}  \label{sec:intro}
The task of job dispatching and load balancing plays a central role in modern communication systems and is critical to ensure good delay and throughput performance. The simple model of a central dispatcher and $N$ parallel servers captures many practical systems like data centers, web server farms, and content delivery systems.

One of the most fundamental load balancing policies is the join-the-shortest-queue (JSQ) policy.
The JSQ policy selects the least backlogged server at any time and is known to be delay optimal under certain assumptions on service times and homogeneity of servers \cite{weber_jsq}. 
However, the JSQ policy requires the dispatcher to know the backlogs of all servers at all times. In most practical systems, server backlogs are not readily available to the dispatcher and it must probe the servers individually to reveal their backlogs. Aggregating the backlog information this way adds a significant communication overhead. This motivates the study of scalable load balancing policies with low control overhead. Consequently, the JSQ-$d$ policy was proposed in \cite{mitz_jsq_d} 
that probes $d$ randomly selected servers instead of probing all the servers. It then selects the shortest of the $d$ randomly selected servers. This reduces the control overhead significantly since $d<N$.
Further, the join-the-idle-queue (JIQ) policy 
which introduced memory at the dispatcher  was proposed in \cite{jiq_1}. In the JIQ policy, the servers send a token back to the dispatcher whenever they become idle. This way, the dispatcher keeps track of the idle servers and can dispatch new jobs to one of the idle servers. This allows further reduction in the communication overhead.


Building upon the JIQ policy's idea of dispatcher having memory, we consider a system in which the dispatcher keeps track of observed server backlogs and makes its dispatching decisions based on past observations.
Specifically, whenever the dispatcher selects a server and dispatches new jobs to it, the selected server sends back its current backlog information as part of the acknowledgment. This setup has a very low control overhead since only the single selected server sends its backlog information.
Our objective is to design efficient load balancing strategies for this partially observed system. We focus on finite queues with maximum backlog $K$ and model each server as a $K+1$ state Markov chain. This problem can be formulated as a Restless Multi-Arm Bandit (RMAB) problem where each server corresponds to a restless arm.
Notice that the arms are restless since they evolve even when they are not selected due to random departures. 

The optimal solution to the RMAB problem involves dynamic programming and suffers from the curse of dimensionality. Therefore, Whittle \cite{whittle} considered the Lagrangian relaxation of RMAB and proposed a heuristic policy called the Whittle index policy. The Whittle index policy is computationally efficient and is shown to be asymptotically optimal \cite{whittle_optimality} under some conditions. However, this low complexity index solution exists only when a special condition called Indexability is satisfied. Whittle index policies have been used  in various fields like age of information, spectrum sharing, and machine maintenance \cite{vishrant_age, spectrum_sharing_main, spectrum_sharing_2, akbarzadeh_machine}.
Further, in \cite{known_queues_glazebrook}, a Whittle index based job dispatcher was proposed when the queue backlogs are known. However, this policy cannot be applied to partially observed systems since backlogs are unknown.

A standard approach to deal with partially observed RMABs is to convert them to fully observed processes by working with beliefs rather than the exact underlying state. For example, such techniques are widely applied in the literature of spectrum sharing \cite{spectrum_sharing_main, spectrum_sharing_2} where, a controller must select from $N$ on/off channels whose true states are unknown. Since the channel states are unknown, the controller makes its decision based on the beliefs or conditional probabilities that the channels are on.
Using this technique, Whittle index policies have been proposed for spectrum sharing in \cite{spectrum_sharing_main, spectrum_sharing_2}. However, in these works, the belief is a scalar (probability that a channel is active), and admits a closed form solution to the relaxed Bellman equations, which can be used to show indexability and derive the Whittle index policy.
Unfortunately, the belief in our system is a $K+1$ dimensional vector since each queue can take values in $\{0,...,K\}$. This makes it difficult to derive closed-form solutions.

In \cite{whittle_condition_nino}, a sufficient condition called Partial Conservation Law (PCL) was derived to verify indexability of RMAB problems with general multi-dimensional state vectors. Even this condition is hard to verify due to the complex nature of our problem.
In \cite{whittle_condition_akbarzadeh}, a sufficient condition was derived that ensures indexability of resetting processes in which the state resets to a fixed state whenever the controller performs an active action.
However, the load balancing problem we consider is not a resetting process without additional assumptions.


In this paper, we consider the problem of load balancing in a system with one dispatcher and $N$ parallel servers.
The dispatcher must select one server at each time-step, but does not know the queue backlogs and must rely on past observations. It receives a reward depending on the selected server's backlog, and its objective is to maximize the total discounted reward with a discount factor $\beta$. Optimal solution to this problem involves dynamic programming and is computationally infeasible. Therefore, we seek to verify Whittle indexability and design a low complexity Whittle index policy. We summarize our contributions below.
\begin{itemize}
	\item We model our job dispatching problem as a RMAB problem where each arm's state is defined as the tuple of  most recent observation and the age of that observation. We consider the following two server models. 
	\item The finite queue model: We derive the closed-form solution to this model's single-arm decoupled problem and show that it has a threshold structure. We then prove its indexability and derive its Whittle index in closed form.
	\item The blocking queue model: We show that the optimal solution to this model's decoupled problem also has a threshold structure. We derive a sufficient condition for indexability of such problems with threshold optimality. We then use this condition to verify indexability for two special cases: $K\leq 2$ and $\beta\leq 0.5$.
\end{itemize}

The rest of the paper is organized as follows. We describe our model and formulate the RMAB problem in Section \ref{sec:prob}. We consider the RMAB's Lagrange relaxation and derive the decoupled single-arm problem in Section \ref{sec:relaxation}. We study the finite queues model in Section \ref{sec:model1} and the blocking queues model in Section \ref{sec:model2}. Finally, we present our numerical results in Section \ref{sec:sims} and conclusions in Section \ref{sec:conc}.

\section{Problem Formulation}  \label{sec:prob}
We consider a job dispatching system with one job dispatcher and $N$ parallel servers. The system operates at discrete time-slots $t \in \{0,1,...\}$. Each server has a first-in-first-out finite queue to buffer incoming jobs and its maximum buffer capacity is $K$. We denote by $Q_i(t)\in\{0,1,...,K\}$ the queue backlog at server $i\in\{1,...,N\}$ at time $t$. The queue service rates are denoted by $\mu_i, \; i\in\{1,...,N\}$, and we model the number of jobs departing server $i$ at time $t$ by $D_i(t)$, a Poisson random variable with rate $\mu_i$ and independent across time. Apart from the jobs sent by the dispatcher, our model also allows exogenous arrivals but we later make a zero exogenous arrival assumption to show threshold optimality.

At each time $t$, the dispatcher chooses exactly one server to dispatch new jobs. We denote the dispatching action at time $t$ by $\bpi(t) := (\pi_1(t), \pi_2(t), ..., \pi_N(t))$, where $\pi_i(t)\in \{0,1\}$ $\forall i$ and $\sum_{i=1}^{N} \pi_i(t) = 1$. Also, we denote the server chosen by dispatcher at time $t$ by $i_\pi(t) \in \{1,...,N\}$ so that $\pi_{i(t)}(t) = 1$. To ease our analysis, we assume that there are at least $K$ jobs available at the dispatcher at all times. We denote by $A_{i,\pi_i}(t)$ the number of new jobs dispatched to server $i$ at time $t$ under action $\pi_i(t)$.

We study two variations of the server model.
In the first variation, we consider a standard finite queues model for each server. The dispatcher's objective is to find the shortest queue and maximize the number of jobs served. Hence it sends as many jobs as possible to the selected server. However, since the selected server has a finite capacity $K$, it only accepts $K- Q_i(t)$, where $Q_i(t)$ is its current backlog. Therefore, we have $A_{i,\pi_i}(t) = \pi_i(t) \left(K - Q_i(t)\right)$ for this model.
In the second variation, we consider a finite blocking queue model for each server. In this model, each server accepts new jobs only if its backlog is zero. Therefore, we have $A_{i,\pi_i}(t) = K \pi_i(t) \Ind[{Q_{i}(t)} = 0]$ for this model, where $\Ind[\cdot]$ is the indicator function.
We will discuss the two variations in detail in Sections \ref{sec:model1} and \ref{sec:model2}.
In summary, including the departures $D_i(t)$ and dispatched jobs $A_{i,\pi_i}(t)$, we can express the queue backlog evolution as $\forall i = 1,...,N$, and $\forall t = 0,1,...$,
\begin{equation}\label{eq:queue_evol}
	Q_i(t+1) = \max\{Q_i(t) + A_{i,\pi_i}(t) - D_i(t), 0\}.
\end{equation}

The dispatcher does not have knowledge of the queue backlogs and makes its decision based on previous observations. Further, the dispatcher observes the backlog $Q_{i_\pi(t)}(t)$ only for the selected server $i_\pi(t)$ at each time $t$. Given server $i$'s history of observations and actions, denoted by $\mathcal{H}_i$, we define its belief vector $\bomega_i(t)$ as its conditional distribution given the history. Formally, $\forall i \in \{1,...,N\}$, we define
$$
	\bomega_i(t) := \left(\Pb[Q_i(t)=q \mid \mathcal{H}_i] \right)_{q=0}^K.
$$

Let $\bm{P}_i$ be server $i$'s probability transition matrix when it is not selected by the dispatcher $\pi_i(t) = 0$. Specifically, $\forall (q,j)\in\{0,...,K\}^2$, the $(q,j)$-th element of $\bm{P}_i$ is
$$
(\bm{P}_i)_{qj} := \Pb[Q_i(t+1) = j  \mid  Q_i(t) = q, \pi_i(t) = 0].
$$
In other words, $\bm{P}_i$ captures server $i$'s restless evolution due to departures and exogenous arrivals (if any) when server $i$ is not selected by the dispatcher. Therefore, $\bm{P}_i$ is dependent only on the distribution of server $i$'s departures and exogenous arrivals. 
Now, we can express the future belief $\bomega_i(t+1)$ recursively in terms of current belief $\bomega_i(t)$ and action $\pi_i(t)$ using the probability transition matrix as follows. If server $i$ is not selected ($\pi_i(t) = 0$), we can update the belief vector as $\bomega_i(t+1) = \bomega_i(t) \cdot \bm{P}_i$. Whereas, if server $i$ is selected ($\pi_i(t) = 1$), the dispatcher observes the backlog $Q_i(t)$ and sends an additional $A_{i,\pi_i}(t)$ jobs to the queue. At this point, we can temporarily set the belief vector to $\bm{\delta} \left(Q_i(t)+A_{i,\pi_i}(t) \right)$ where $\bm{\delta}(q)$ is a one-hot vector with one at the $q$-th index. However, we also need to account for the departures and  exogenous arrivals. Accounting for those, we can write the final belief vector update as $\bomega_i(t+1) = \bm{\delta}\left(Q_i(t)+A_{i,\pi_i}(t) \right) \cdot \bm{P}_i$ if $\pi_i(t) = 1$. In summary, we have
\begin{equation}\label{eq:belief_update}
\bomega_i(t+1) = \begin{cases}
	\bomega_i(t) \cdot \bm{P}_i & \text{if } \pi_i(t) = 0, \\
	\bm{\delta}\left(Q_i(t)+A_{i,\pi_i}(t) \right) \cdot \bm{P}_i & \text{if } \pi_i(t) = 1.
\end{cases}
\end{equation}

It has been previously shown \cite{belief_sufficient_stat} that the collection of belief vectors $\{\bomega_i(t), i\in\{1,...,N\}\}$ is a sufficient statistic for this problem. Hence, we focus on policies that take action $ \bpi(t)$ only as a function of the belief vectors $\bomega_i(t)$. However, each $\bomega_i(t)$ is a $K+1$ dimensional vector, and there are $N$ such vectors corresponding to $i=1,...,N$. To simplify this, we consider the following alternate representation of the belief vector. For each server $i$, we keep track of the most recent backlog observation $q_i(t)$ and the time since this observation $\tau_i(t)$. Notice that the belief vector can be recovered from these variables at any time using $\bomega_i(t) = \bm{\delta}(q_i(t)) \cdot \bm{P}_i^{\tau_i(t)}$. Further, we can update $q_i(t)$ and $\tau_i(t)$ recursively as
\begin{equation}\label{eq:age_update}
	\tau_i(t+1) = \begin{cases}
		\tau_i(t) + 1 & \text{if } \pi_i(t) = 0, \\
		1 & \text{if } \pi_i(t) = 1.
	\end{cases}
\end{equation}
\begin{equation}\label{eq:obs_update}
	q_i(t+1) = \begin{cases}
		q_i(t) & \text{if } \pi_i(t) = 0, \\
		Q_i(t)+A_{i,\pi_i}(t) & \text{if } \pi_i(t) = 1.
	\end{cases}
\end{equation}
Therefore, we define the state  $S_i(t)$ of server $i$ at time $t$ as the tuple of most recent observation and the time since that observation, $S_i(t) := (q_i(t), \tau_i(t))$. Also, we denote the state of the entire system by $\bm{S}(t) := (S_1(t), ..., S_N(t))$. Again, notice that the server states evolve even when they are not selected $\pi_i(t) = 0$. In other words, the servers are restless and hence, this is a restless multi-arm bandit (RMAB) problem.

At each time $t$, the dispatcher receives a reward equal to the number of new jobs dispatched,
$$R_\pi(t) := \sum_{i=1}^N A_{i,\pi_i}(t) = \sum_{i=1}^N \pi_i(t) A_{i,\pi_i}(t).$$
where, we explicitly show the dependency on $\pi_i(t)\in\{0,1\}$ to highlight that the dispatcher selects only one server.
Notice that the dispatcher must select the shortest queue to maximize this reward, since selecting the shortest queue allows it to dispatch the maximum number of jobs.
Our objective is to design a dispatcher that maximizes the expected total discounted reward given the initial state $\bm{S}(0)$. Formally, let $\Pi$ be the collection of all policies and $\beta\in(0,1)$ be the discount factor, our objective is
\begin{gather*}
	\max_{\pi\in\Pi} \Ex_\pi \left[ \sum_{t=0}^\infty \beta^t \sum_{i=1}^N \pi_i(t) A_{i,\pi_i}(t) \biggm| \bm{S(0)} \right] \\
	\text{such that }\sum_{i=1}^N \pi_i(t) = 1, \; \forall t=0,1,...
\end{gather*}

\section{Lagrangian Relaxation and the Single-Arm Decoupled Problem}\label{sec:relaxation}

Obtaining the optimal solution to the above RMAB optimization problem involves dynamic programming and is computationally infeasible. Whittle \cite{whittle} considered the Lagrangian relaxation of the RMAB problem and decoupled the arms resulting in $N$ independent optimization problems. We consider the following relaxation of our problem.
\begin{gather*}
	\max_{\pi\in\Pi} \Ex_\pi \left[ \sum_{t=0}^\infty \beta^t \sum_{i=1}^N \pi_i(t) A_{i,\pi_i}(t) \mid \bm{S(0)} \right] \\
	\text{such that }\Ex_\pi \left[  \sum_{t=0}^\infty \beta^t \sum_{i=1}^N \left(1-\pi_i(t)\right) \mid \bm{S(0)} \right] = \frac{N-1}{1-\beta}
\end{gather*}

Using Lagrange multiplier $\gamma$, we have
\begin{equation*}
	\max_{\pi\in\Pi} \Ex_\pi \left[ \sum_{t=0}^\infty \beta^t \sum_{i=1}^N \left[ \pi_i(t) A_{i,\pi_i}(t) + \gamma \left(1-\pi_i(t)\right) \right] \mid \bm{S(0)} \right]
\end{equation*}


We can decouple the above optimization into $N$ independent single-arm problems. Dropping the subscript $i$, we can write the single-arm decoupled problem as
\begin{equation*}
	\max_{\pi\in\Pi} \Ex_\pi \left[ \sum_{t=0}^\infty \beta^t \left[\pi(t) A_\pi(Q(t)) + \gamma \left(1-\pi(t)\right)\right] \mid S(0) \right]
\end{equation*}
where, we write the number of dispatched jobs as $ A_\pi(Q(t))$ to emphasize its dependency on the current backlog $Q(t)$.
We can interpret this single arm decoupled problem as follows. Given the single-arm state $S(t) = (q, \tau)$, the dispatcher can either probe the server $\pi(t) = 1$ or stay idle $\pi(t) = 0$. If the dispatcher probes, it observes the current backlog $Q(t)$, sends $A_\pi(Q(t))$ packets, updates its state to  $S(t+1) = \left( Q(t) + A_\pi(Q(t)), 1 \right)$, and receives a reward $A_\pi(Q(t))$. Whereas if the dispatcher stays idle, it does not observe the backlog, updates its state to $S(t+1) = (q, \tau+1)$, and receives a constant subsidy $\gamma$. Let $V_\gamma(q,\tau)$ be the value function for state $(q,\tau)$ i.e. the optimal expected discounted reward received starting from an initial state $S(0) = (q,\tau)$. Using Bellman equations, we can express the value function as, $\forall q\in\{0,...,K\}$ and $\forall \tau\in\{1,2,...\}$,  
\begin{multline}\label{eq:value_func_general}
	V_\gamma(q,\tau) = \max \Bigg\{\gamma + \beta 	V_\gamma(q,\tau+1), \\
		\sum_{q'=0}^K P_{qq'}(\tau) \left(A_\pi(q') + \beta V_\gamma(q'+A_\pi(q'),1) \right) \Bigg\}
\end{multline}
where, $P_{qq'}(\tau)$ is the probability that the current backlog is $q'$ given that the backlog was $q$ some $\tau$ time-steps ago. This can be obtained from the probability transition matrix as $P_{qq'}(\tau):= (\bm{P}^\tau)_{qq'}$. Notice that the value function is given by the maximum of the values received by idling and probing. When the dispatcher idles, it receives a constant subsidy $\gamma$, and updates its state to $(q,\tau+1)$. Whereas when it probes, it finds the backlog at $q'$ with probability $P_{qq'}(\tau)$, receives a reward $A_\pi(q')$, and updates its state to $\left( q'+A_\pi(q'), 1 \right)$.


After relaxing the RMAB problem and decoupling the arms, Whittle \cite{whittle} proposed the heuristic Whittle Index policy whenever the problem is \textit{indexable}. The indexability property states that as the subsidy $\gamma$ increases from $0$ to $\infty$, the set of states for which it is optimal to idle increases monotonically from the empty set $\phi$ to the entire state space. Further, given indexability, the \textit{Whittle index} for state $S$ is defined as the lowest subsidy $\gamma$ that makes it equally desirable to activate and idle in state $S$.
In the next two sections, we consider two queueing models and study their indexability. Our approach involves showing that the optimal solution to the single-arm decoupled problem has a threshold structure and then verifying that the threshold values are monotonic with the subsidy. However, to show threshold optimality, we require the transition matrix $\bm{P}$ to be lower triangular. Hence, we assume that there are no exogenous arrivals to the servers, and the only arrivals to the server are sent from the dispatcher.
\begin{assumption}\label{assump:lower_triangular_P}
	The transition matrix $\bm{P}$ is lower triangular.
\end{assumption}

\section{Model 1: Finite Queues}  \label{sec:model1}
In this model, the server has a finite queue with maximum backlog capacity $K$. Since the dispatcher's objective is to find the shortest queue and maximize the number of jobs served, it sends as many jobs to the selected server  as possible. However, since the server is finite with capacity $K$, the number of new jobs dispatched is
$$A_\pi(Q(t)) =  \pi(t) \left(K - Q(t)\right).$$
In other words, whenever the dispatcher probes the server, it dispatches $K - Q(t)$ new jobs, thereby filling up the queue to $K$. 
This is possible due to our assumption that there are at least $K$ jobs available at the dispatcher at all times.
Further, since the dispatcher always fills up the selected queue, the last seen backlog is always $q(t) = K$. Hence, the system state is always $S(t) = (K, \tau(t))$ at all times $t$, where $\tau(t)$ is the time since last filled. Since $q(t)=K$ is always fixed, it suffices to keep track of only $\tau(t)$. This essentially makes the system state one-dimensional and allows us to solve the Bellman equation in closed-form.
We start by simplifying the value function \eqref{eq:value_func_general} by plugging in this model's $A_\pi(Q(t))$ as, $\forall \tau\in\{1,2,...\}$,
\begin{multline} \label{eq:value_func_model1}
		V_\gamma(\tau) = \max \Bigg\{\gamma + \beta V_\gamma(\tau+1), \\ \sum_{q'=0}^K P_{Kq'}(\tau) \left(K - q' + \beta V_\gamma(1)\right)  \Bigg\}
\end{multline}
where, we have dropped the value function's dependency on $q$ since $q(t)=K$ for all $t$. We now show that the optimal solution $\pi^*(\tau)$ to this value function has a threshold structure in $\tau$. We denote the conditional expected queue backlog given $\tau$ (the time since last filled) by $E(\tau):= \sum_{q=0}^K q \cdot P_{Kq}(\tau)$ with $E(0)=K$. This expectation has the following properties.
\begin{lemma}\label{lemma:E_monotone}
For all $\tau=1,2,...$, $E(\tau)$ is non-increasing in $\tau$ and $E(\tau) - E(\tau-1)$ is non-increasing in $\tau$.
\end{lemma}

This lemma can be proven using the fact that the expected backlog $E(\tau)$ is convex in $\tau$ since the departures are non-negative.
\ifisreport
We show the proof in Appendix \ref{proof:monotonicity_of_E}.
\fi
Notice that since the activation reward $K-Q(t)\in[0,K]$, if the subsidy $\gamma>K$, it is optimal to always idle and similarly if the subsidy $\gamma<0$, it is optimal to always probe. Therefore, we focus on the non-trivial case $\gamma \in [0,K]$ for the rest of the analysis. We now present the threshold optimality result.
\begin{theorem}\label{thm:model1_threshold}
	For any $\tau \in\{1,2,...\}$, the optimal action is
	\begin{equation*}
		\pi^*(\tau) = \begin{cases}
			1 & \text{ for } \tau \geq \tau^*,\\
			0 & \text{ for } \tau < \tau^*
		\end{cases}
	\end{equation*}
	where, $\tau^*\in\{1,2,...\}$ is such that $W(\tau^*) \leq \gamma \leq W(\tau^* + 1)$ and $W:\{1,2,...\}\rightarrow[0,K]$ is defined as
	\begin{equation*}
		W(\tau) := K - \left(\frac{1-\beta^\tau}{1-\beta}\right) E(\tau-1) + \beta \left(\frac{1-\beta^{\tau-1}}{1-\beta}\right)E(\tau).
	\end{equation*}
\end{theorem}

The proof involves showing that the value function under the policy $\pi^*(\tau)$ given in theorem's statement satisfies the Bellman equations \eqref{eq:value_func_model1}.
\ifisreport
We show the proof in Appendix \ref{proof:model_1_threshold}.
\fi
This theorem states that, once we probe the server, it is optimal to wait a fixed $\tau^*$ amount of time before probing it again. Next, we present the Whittle indexability for this model.
\begin{theorem}
	The job dispatching problem under the standard finite queues model (Model 1) is indexable.
\end{theorem}

\begin{IEEEproof}
	We start by showing that $W(\tau)$ is non-decreasing in $\tau\in\{1,2,...\}$. Evaluating $W(\tau+1)-W(\tau)$ from $W(\tau)$'s definition, simplifying it, and using Lemma \ref{lemma:E_monotone}, we have $\forall \tau$,
	\begin{align*}
		W(\tau&+1)-W(\tau)\\
		&= \frac{1-\beta^\tau}{1-\beta} [ \{E(\tau-1) - E(\tau)\} - \beta \{E(\tau) - E(\tau+1)\} ] \\
		& \geq (1-\beta^\tau) [E(\tau) - E(\tau+1)] \geq 0.
	\end{align*}

	Now, from Theorem \ref{thm:model1_threshold}, we know that $\tau^* \in\{1,2,...\}$ and $W(\tau^*) \leq \gamma \leq W(\tau^* + 1).$ Hence, due to $W$'s monotonicity, we have that the threshold $\tau^*$ is non-decreasing in $\gamma$. This implies Whittle indexability due to threshold optimality.
\end{IEEEproof}

Finally, notice that model 1's process is resetting i.e. under the active action (dispatcher probes the server), the state always resets to $(K,1)$. Therefore, previous techniques \cite{whittle_condition_akbarzadeh} can be applied to verify indexability alternatively. However, our method above derives the Whittle index in closed form and expresses it in terms of the expected values.

\section{Model 2: Blocking Finite Queues}  \label{sec:model2}
In this model, the server is blocking which means it accepts new jobs only if it is empty. Therefore, the number of new jobs dispatched is
$$A_\pi(Q(t)) := K \Ind[Q(t) = 0].$$

At time $t$, if the dispatcher probes the server and finds it empty, it sends $K$ new jobs, receives a reward of $K$, and updates its state to $(K,1)$. Whereas if it finds the server non-empty, it does not receive any reward, but still gets to observe the current backlog (say $q'$), and updates its state to $(q',1)$. Plugging in this model's $A_\pi(Q(t))$, we can simplify the value function \eqref{eq:value_func_general} as, $\forall q\in\{0,...,K\}$ and $\forall \tau\in\{1,2,...\}$,
\begin{equation}\label{eq:value_func_model2}
	V_\gamma (q,\tau) = \max \left\{\hspace{-0.5mm} \gamma + \beta V_\gamma (q, \tau+1), \sum_{q'=0}^{K} P_{qq'}(\tau) \wtilde{V}_\gamma(q') \hspace{-0.5mm} \right\}
\end{equation}
where, we use the following notation for convenience
\begin{equation*}
\wtilde{V}_\gamma(q') := \begin{cases}
	\beta V_\gamma(q',1) & \text{ for }q'\in\{1,...,K\},\\
	K + \beta V_\gamma(K,1) &  \text{ for }q'=0.
\end{cases}
\end{equation*}
In the rest of this section, we show that the optimal solution to this problem is a threshold policy in $\tau$ and derive indexability of this model for some special cases.


\subsection{Threshold Optimality}
Models similar to the single-arm version of the blocking queues model have been studied previously in the machine repair literature. \cite{rosenfield_machine, rosenfield_machine_2} study the optimal policy to maintain a deteriorating machine under available actions: idle, inspect, and repair. \cite{levin_machine} studies a similar problem under idle and terminate actions for finite time horizon. It has been shown that the optimal policy is switch-type in terms of the time since last observation. We use some results from these works to show threshold optimality of the blocking queues model. Specifically, we use the following lemmas from \cite{rosenfield_machine, levin_machine}.
\begin{lemma}[Single Crossing \cite{rosenfield_machine}] \label{lemma:rosenfield}
	Consider a sequence $x_j$, $j=0,...,K$ such that $x_j$ crosses zero only once
	and it does so from above (i.e.,  $\exists$ $j'$ such that $x_{j}>0$ for $j\leq j'$ and $x_{j}<0$ for $j>j'$).
	Denote $G(q,\tau):=\sum_{q'=0}^K P_{qq'}(\tau) x_{q'}$. Under Assumption \ref{assump:lower_triangular_P} on $\bm{P}$,  for any $(q,\tau)$,
	\begin{enumerate}
		\item if $G(q,\tau)>0$, then $G(q,\tau')>0$ for all $\tau'\geq\tau$,
		\item if $G(q,\tau)<0$, then $G(q',\tau)<0$ for all $q'\geq q$.
	\end{enumerate}
\end{lemma}

This lemma is due to the total positivity property of the Poisson probability transition matrix. See \cite{rosenfield_machine} for proof.

\begin{lemma}[Information Inequality \cite{levin_machine}] \label{lemma:levin} $\forall q\in\{0,...,K\}$,
	$$V_\gamma(q,2) \leq \sum_{q'=0}^K P_{qq'}(1)V_\gamma(q',1)$$
\end{lemma}

This lemma has a rather intuitive interpretation. Consider a situation where the backlog was observed to be $q$ two time-steps ago. The lemma says that the total value obtained can only improve if an oracle were to reveal the intermediate backlog one time-step ago. See \cite{levin_machine} (Lemma 3.2.1) for proof.

We now proceed to show threshold optimality.
Recall that when the action is idle, the state has a fixed transition from $(q, \tau)$ to $(q, \tau+1)$ since there is no observation involved. This allows us to reinterpret the action as the number of time-steps to wait before probing again. Denote by $L_{(q,\tau)}(l)$ the total reward obtained starting from state $(q, \tau)$, if we idle for the first $l$ time-steps, probe at $(l+1)$-th time-step, and continue optimally. We have $\forall q,\forall\tau$, and $\forall l\in\{0,1,...\}$,
\begin{equation*}
	L_{(q,\tau)}(l) := \gamma \left(\frac{1-\beta^l}{1-\beta}\right) + \beta^l \sum_{q'=0}^K P_{qq'}(\tau+l) \wtilde{V}_\gamma(q')
\end{equation*}

The optimal value function $V(q,\tau)$ and the optimal action $\pi^*(q,\tau)$ can be recovered from $L_{(q,\tau)}(l)$ as
\begin{align*}
	V(q,\tau) &= \max_{l\geq 0} L_{(q,\tau)}(l) \\
	\pi^*(q,\tau) &= \Ind\left[\arg\max_{l\geq 0} L_{(q,\tau)}(l) =0\right]
\end{align*}
where, we break ties by selecting the smallest index.  We now show an important structural property of $L_{(q,\tau)}(l)$.
\begin{proposition}\label{prop:L_structure}
	If $L_{(q,\tau)}(l) > L_{(q,\tau)}(l+1)$ for some $l\in\{0,1...\}$, then $\forall k \geq l$, we have $L_{(q,\tau)}(k) > L_{(q,\tau)}(k+1)$. Further, if $L_{(q,\tau)}(l) < L_{(q,\tau)}(l+1)$ for some $q\in\{1...,K\}$, then $\forall q' \geq q$, we have $L_{(q',\tau)}(l) < L_{(q',\tau)}(l+1)$
\end{proposition}

We prove this in Appendix \ref{proof:L_structure}.
This proposition states that once $L_{(q,\tau)}(l)$ starts decreasing, it will continue to decrease. In other words, if we have already waited long enough (say $l_0$ time-steps) such that $L_{(q,\tau)}(l_0)$ is decreasing, then we are better off probing immediately, since waiting more will only reduce the value. This in turn implies that there is an optimal amount of time-steps to wait after which probing is optimal.
We now show this formally in the following theorem.
\begin{theorem}\label{thm:model2_threshold}
	There is an optimal threshold value $\tau_q^* \geq 1$ for each observation $q\in\{0,...,K\}$ such that $\forall \tau\in\{1,2,...\}$, the optimal action is given by
	\begin{equation*}
		\pi^*(q,\tau) = \begin{cases}
			1 & \text{ for } \tau \geq \tau_q^*,\\
			0 & \text{ for } \tau < \tau_q^*.
		\end{cases}
	\end{equation*}
	Further, the optimal thresholds are such that $\tau_1^*\leq\tau_2^*...\leq\tau_K^*$.
\end{theorem}

\ifisreport
We prove this in Appendix \ref{proof:model_2_threshold}.
\fi
Theorem \ref{thm:model2_threshold} states that once we observe the queue backlog to be at $q$, we must idle for $\tau_q^*$ time-steps before probing again. Further, we must wait longer if the observed backlog is higher
$\tau_{q}^* \leq \tau_{q'}^*$ if $q \leq q'$. This is as expected because more pending jobs in the queue implies longer time to empty, hence we must wait longer before we send our next probe. Finally, even though we have shown threshold optimality, we are unable to derive the threshold values in closed-form due to the two-dimensional nature of the state.

\subsection{Indexability for special cases}
Since the optimal policy has a threshold structure, we can show indexability by verifying that the thresholds $\tau_q^*$'s are non-decreasing with $\gamma$. However, since we were unable to derive the threshold values in closed-form, proving indexability for this model is more difficult than the standard queue model.

 In \cite{spectrum_sharing_main}, a sufficient condition for indexability was derived in terms of total discounted passive time. However, this condition was derived for the case of scalar beliefs. We now extend this result to our two-dimensional $(q,\tau)$ state-space.
 We define the total discounted active time of a policy $\pi$ given initial state $S(0)=(q,\tau)$ as
\begin{equation*}
	A_\pi(q,\tau) := \Ex_\pi \left[ \sum_{t=0}^\infty \beta^t \pi(t) \mid S(0)=(q,\tau) \right]
\end{equation*}

In the following proposition, we derive a sufficient condition for indexability in terms of the active time $A_\pi(q,\tau)$. Let $\pi_0$ be a policy that probes at $t=0$ and continues optimally. Similarly, let $\pi_1$ be a policy that idles at $t=0$ and continues optimally. In the proposition, we highlight the dependency of threshold values $\tau_q^*$ on subsidy $\gamma$ by explicitly writing $\tau_q^*(\gamma)$.
\begin{proposition}\label{prop:active_condition}
	The optimal threshold values $\tau_q^*(\gamma)$ are non-decreasing in $\gamma$ if $\forall q\in\{0,...,K\}$, $A_{\pi_1}(q,\tau_q^*) \geq A_{\pi_0}(q,\tau_q^*)$.
\end{proposition}

See proof in Appendix \ref{proof:active_condition}.
Notice that this condition is similar to the total discounted passive time condition in \cite{spectrum_sharing_main}. We have extended this condition to our two-dimensional and discrete state $(q,\tau)$.
A similar condition was also derived in \cite{xiaojun_age} which needs verifying for a sub-class of policies but for all states. Our condition is different in that it needs verifying specifically for our threshold optimal policy, but only at the decision boundaries $(q,\tau^*_q)$.
Next, we use Proposition \ref{prop:active_condition}'s condition to check indexability. However, since we do not have a closed-form expressions for this model, checking this condition is also difficult. Moreover, this model is not resetting i.e. the active action does not reset the state to a fixed state.
Therefore, previously developed techniques for verifying indexability of resetting processes cannot be used.
Consequently, verifying general indexability is difficult. Therefore, we show indexability for two special cases: $K\leq 2$ and $\beta \leq 0.5$. And for general values of $K$ and $\beta$, we show some numerical trends that support indexability in Section \ref{sec:sims}. Notice that our first special case $K\leq 2$ is the first step towards a general indexability result since $K=2$ is the smallest value of $K$ such that the belief is a vector. 
\begin{theorem} \label{thm:indexability_smallK}
	The job dispatching problem under the blocking queues model (Model 2) is indexable for $K \leq 2$.
\end{theorem}

See proof in Appendix \ref{proof:indexability_smallK}. The second special case $\beta\leq 0.5$ corresponds to settings where instantaneous rewards carry more weight than future rewards. 

\begin{theorem} \label{thm:indexability_smallbeta}
	The job dispatching problem under the blocking queues model (Model 2) is indexable for $\beta \leq 0.5$.
\end{theorem}
\begin{IEEEproof}
	Since policy $\pi_1$'s action at $t=0$ is probe, we have $A_{\pi_1}(q,\tau_q^*) = 1+ \Ex_{\pi_1} \left[ \sum_{t=1}^\infty \beta^t \pi_1(t) \mid S(0)=(q,\tau) \right] \geq 1.$ Similarly, since policy $\pi_0$ idles at $t=0$, we have $A_{\pi_0}(q,\tau_q^*) = \Ex_{\pi_0} \left[ \sum_{t=1}^\infty \beta^t \pi_0(t) \mid S(0)=(q,\tau) \right] \leq \beta/(1-\beta).$ Therefore, $A_{\pi_1}(q,\tau_q^*) - A_{\pi_0}(q,\tau_q^*) \geq (1-2\beta)/(1-\beta) \geq 0$ using $\beta\leq0.5$. Hence, we have indexability due to Proposition \ref{prop:active_condition}.
\end{IEEEproof}

In Theorems \ref{thm:indexability_smallK} and \ref{thm:indexability_smallbeta}, we showed indexability for two special cases. In Section \ref{sec:model2_sims}, we conduct numerical simulations for a general case and show evidence of general indexability.

\section{Simulations}  \label{sec:sims}
\subsection{Model 1: Finite Queues Model}\label{sec:model1_sims}
We simulate the finite queues model with discount factor $\beta=0.99$, $N=10$ parallel servers each with capacity $K=6$.
\ifisreport
Fig. \ref{fig:model1_avg_reward_} shows the total reward  $\mathcal{R}_\pi(T):=\sum_{t=0}^T \beta^t \Ex [R_\pi(t)]$ of the three policies as a function of $T$.
Fig. \ref{fig:model1_gap} shows the performance gap of the policies relative to oracle $\left({\mathcal{R}_{JSQ}(T) - \mathcal{R}_\pi(T)}\right) / {\mathcal{R}_{JSQ}(T)}$.
\begin{figure}[hbtp]
	\centering
	\begin{subfigure}[h]{0.8\linewidth}
		\centering
		\includegraphics[width=\linewidth]{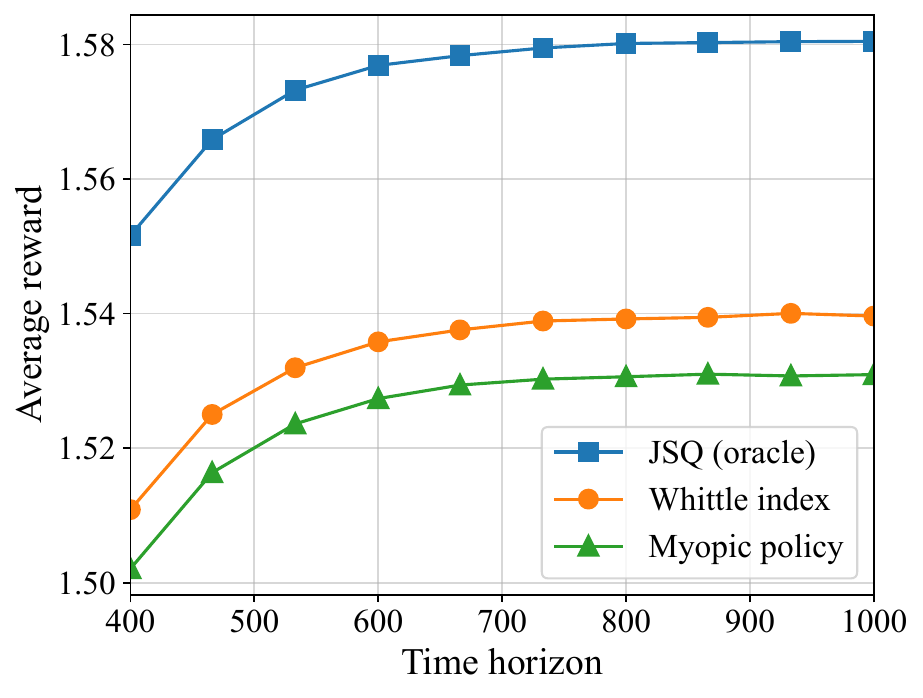}\vspace*{-0.2cm}
		\caption{Average reward.}\label{fig:model1_avg_reward_}\vspace*{0.25cm}
	\end{subfigure}

	\begin{subfigure}[h]{0.8\linewidth}
		\centering
		\includegraphics[width=\linewidth]{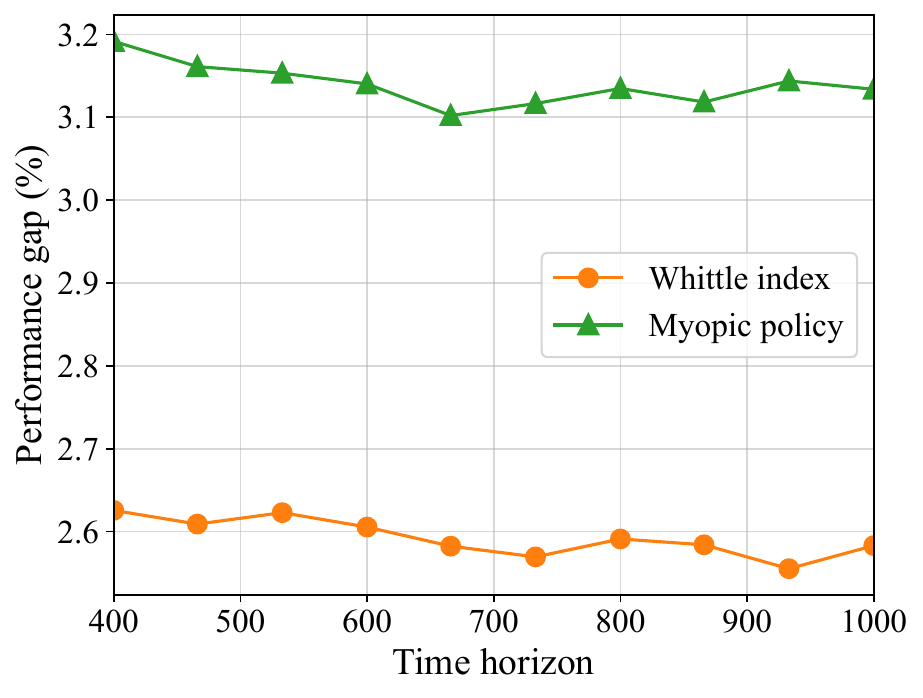}\vspace*{-0.2cm}
		\caption{Performance gap.}\label{fig:model1_gap}
	\end{subfigure}
	\caption{Reward performance for the finite queues model.}\vspace{-0.1cm}
	\label{fig:model1_reward}
\end{figure}
\fi
We pick the following variable server rates, $\mu_i=0.05$ for $i=1,...4$, $\mu_i=0.1$ for $i=5,...,8$, and $\mu_i=0.5$ for $i=9,10$. We simulate the Whittle index policy, myopic policy, and the oracle JSQ policy. The myopic policy greedily picks the queue with least expected backlog at every time-step. The oracle JSQ policy has knowledge of the backlogs and picks the shortest queue.
As can be seen, the Whittle index policy has higher reward than the myopic policy, and its performance is comparable to the oracle JSQ.
\subsection{Model 2: Blocking Queues Model}\label{sec:model2_sims}
We first conduct numerical analysis of the single-arm decoupled problem for this model. For a given subsidy $\gamma\in[0,K]$, we obtain the optimal thresholds $\tau_q^*(\gamma)$'s numerically from value iteration on Bellman equations in \eqref{eq:value_func_model2}. We repeat this calculation for various values of $\gamma\in[0,K]$ and plot the optimal thresholds $\tau_q^*(\gamma)$ as a function of $\gamma$.
\begin{figure*}[hbtp]
	\centering
	\begin{subfigure}[h]{0.32\textwidth}
		\centering
		\includegraphics[width=\textwidth]{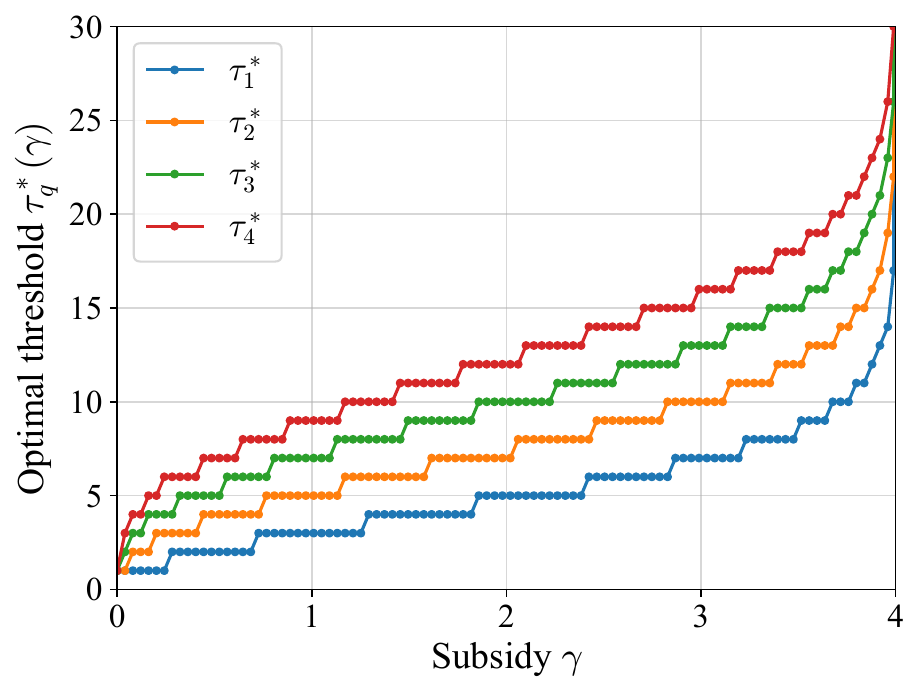}
		\caption{Optimal thresholds $\tau^*_q(\gamma)$ for $K=4$.}\label{fig:model2_indexability_K4}
	\end{subfigure}
	\begin{subfigure}[h]{0.32\textwidth}
		\centering
		\includegraphics[width=\textwidth]{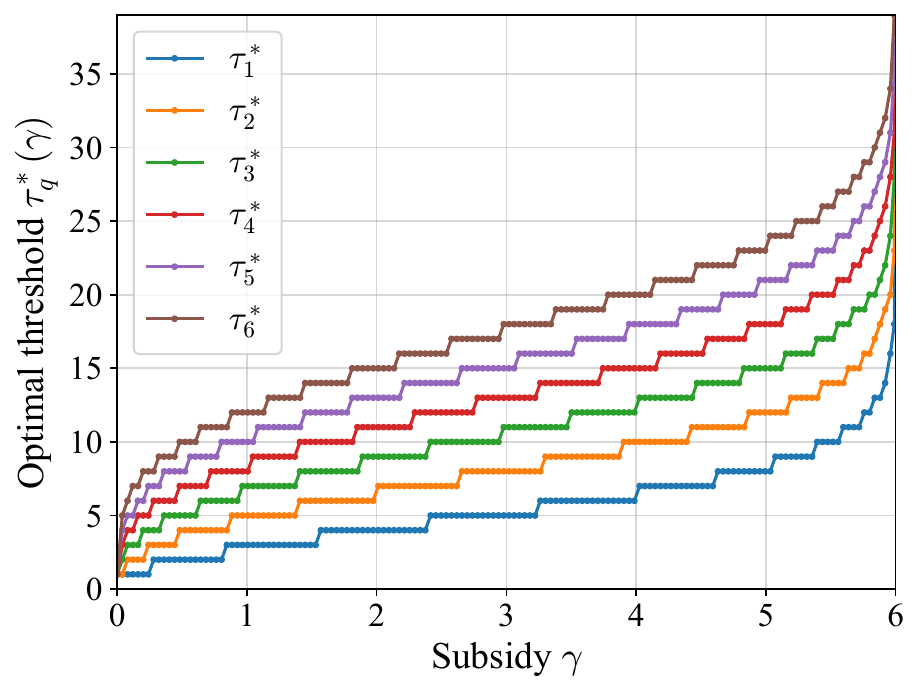}
		\caption{Optimal thresholds $\tau^*_q(\gamma)$ for $K=6$.}\label{fig:model2_indexability_K6}
	\end{subfigure}
	\begin{subfigure}[h]{0.32\textwidth}
		\centering
		\includegraphics[width=\textwidth]{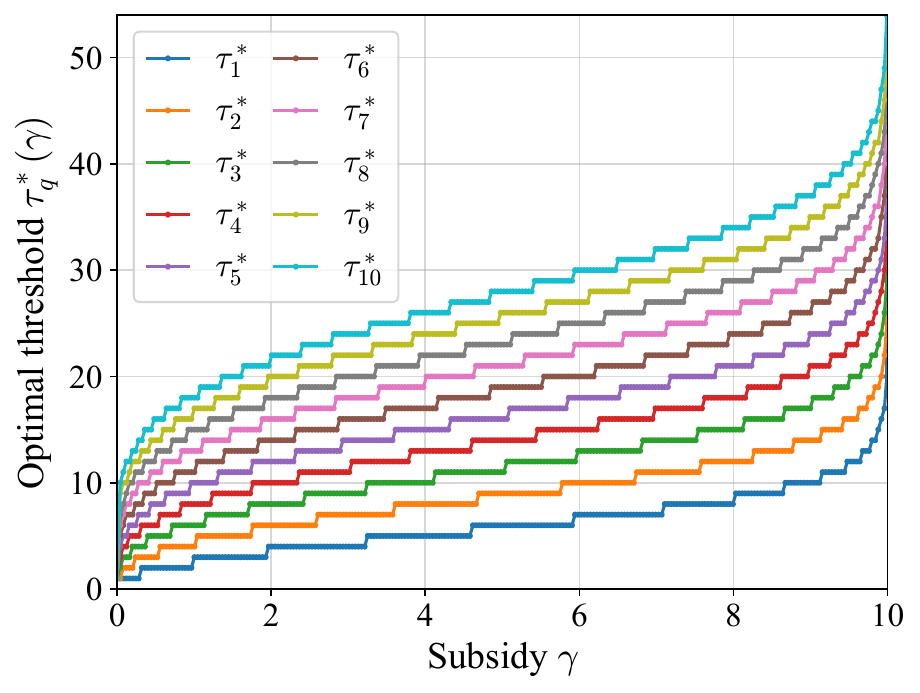}
		\caption{Optimal thresholds $\tau^*_q(\gamma)$ for $K=10$.}\label{fig:model2_indexability_K10}
	\end{subfigure}
	\caption{Optimal thresholds $\tau^*_q(\gamma)$ as a function of $\gamma$ for the blocking queues model's single-arm decoupled problem.}\vspace{-0.35cm}
	\label{fig:model2_indexability}
\end{figure*}
Fig. \ref{fig:model2_indexability} shows the plots of thresholds $\tau_q^*(\gamma)$ against subsidy $\gamma$.
The figure contains the plot for three values of $K=4,6,$ and $10$. We use $\beta=0.99$ and $\mu=0.5$. As can be seen, for all values of $q\in\{1,...,K\}$, the optimal threshold $\tau_q^*(\gamma)$ starts from $1$ and monotonically increases to $\infty$ as subsidy $\gamma$ increases. This shows that the passive set is monotonically non-decreasing with subsidy indicating indexability for general values of $K$ and $\beta$.
Moreover, we can see that the threshold values $\tau_q^*$ are non-decreasing in $q$ as expected since it is optimal to wait longer if the observed backlog is larger and vice-versa.

Next, we simulate the Whittle index, myopic, and the oracle JSQ policies with discount factor $\beta=0.99$ and $N=10$ parallel servers each with $K=6$. We pick variable server rates, $\mu_i=0.05$ for $i=1,...5$, $\mu_i=0.1$ for $i=6,...,9$, and $\mu_i=0.5$ for $i=10$. We compute the approximate Whittle index of each state using the numerical analysis of the single-arm problem described above. However, notice that algorithms like the adaptive greedy algorithm from \cite{whittle_condition_akbarzadeh, xiaojun_age} can be used to compute Whittle index exactly. The myopic policy picks the queue that is most likely to be empty, and the oracle JSQ policy knows the exact backlogs and picks the shortest queue.
\ifisreport
\begin{figure}[hbtp]
	\centering
	\begin{subfigure}[h]{0.8\linewidth}
			\centering
			\includegraphics[width=\linewidth]{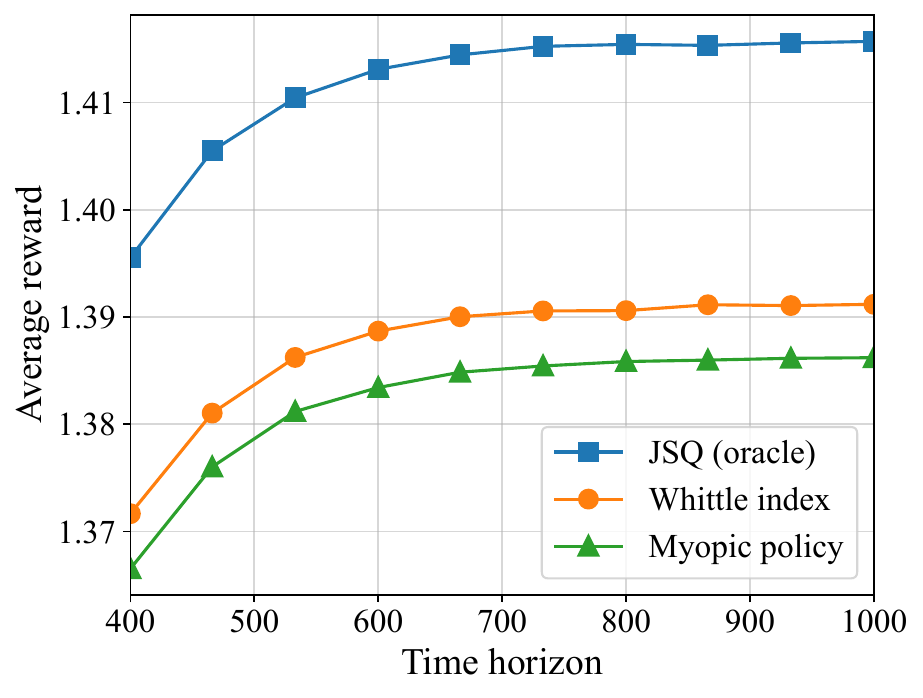}\vspace*{-0.2cm}
			\caption{Average reward.}\label{fig:model2_avg_reward_}\vspace*{0.25cm}
		\end{subfigure}
	\begin{subfigure}[h]{0.8\linewidth}
			\centering
			\includegraphics[width=\linewidth]{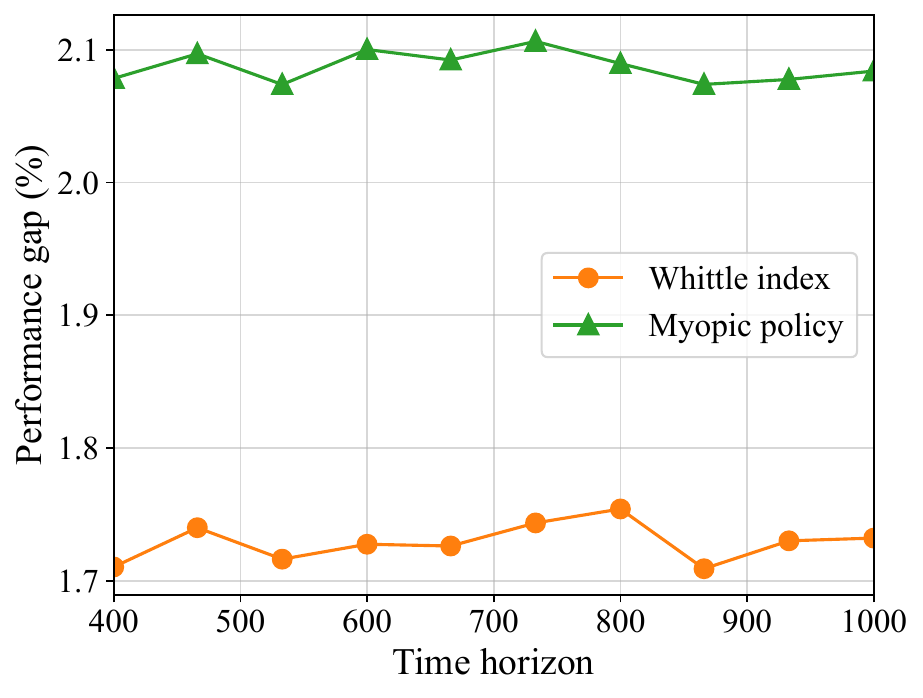}\vspace*{-0.2cm}
			\caption{Performance gap.}\label{fig:model2_gap}
		\end{subfigure}
	\caption{Reward performance for the blocking queues model.}\vspace{-0.1cm}
	\label{fig:model2_reward}
\end{figure}
Fig. \ref{fig:model2_avg_reward_} shows the total discounted reward $\mathcal{R}_\pi(T):=\sum_{t=0}^T \beta^t \Ex [R_\pi(t)]$ as a function of $T$.
Fig. \ref{fig:model2_gap} shows the performance gap relative to oracle $\left({\mathcal{R}_{JSQ}(T) - \mathcal{R}_\pi(T)}\right) / {\mathcal{R}_{JSQ}(T)}$.
\fi
As can be seen, the Whittle index policy beats the myopic policy and is competitive with the oracle JSQ policy.

\section{Conclusion} \label{sec:conc}
We considered the load balancing problem when the server backlogs are partially observed. We formulated this problem as RMAB with arm state defined as the tuple of most recent observation and the time since the recent observation. For the standard finite queues model, we showed threshold optimality and indexability by deriving closed-form expressions. For the blocking queues model, we showed threshold optimality and proved indexability for some special cases. Some potential future directions to make the model more practical is allowing random arrivals to the dispatcher and allowing exogenous arrivals to the servers. We believe these extensions will be more difficult since the transition matrix is not lower triangular and the value function may not be well behaved.

\appendices
\ifisreport
\section{Proof of Lemma \ref{lemma:E_monotone}}\label{proof:monotonicity_of_E}
Recall that $E(\tau)$ is the conditional expected queue backlog given that the time since last filled is $\tau$. Hence, we have $E(\tau) := \Ex[Q(t) \; | \; Q(t-\tau)=K \text{ and zero new arrivals}]$. Denote the offered service at time $t$ by $D(t) \sim \text{i.i.d. Poisson}(\mu)$. Therefore, we have $E(\tau) - E(\tau+1) = \tilde\Ex[\min\{Q(t), D(t)\}]$ where $\tilde\Ex[\cdot] := \Ex[\cdot | Q(t-\tau)=K]$. Since $Q(t) \geq 0$ and $D(t) \geq 0$, we have $E(\tau) - E(\tau+1) \geq 0$. Moreover, since $D(t)$ is i.i.d. across $t$ and since $Q(t-1) \geq Q(t)$, we have  $\Ex[\min\{Q(t-1), D(t-1)\}] \geq \Ex[\min\{Q(t), D(t)\}]$. Therefore, we have $E(\tau-1) - E(\tau) \geq E(\tau) - E(\tau+1)$ for all $\tau$. \hfill \IEEEQEDhere
\section{Proof of Theorem \ref{thm:model1_threshold}}\label{proof:model_1_threshold}
We start by calculating the value function under the theorem's policy $\pi^*(\tau) = \Ind[\tau \geq \tau^*]$. For notational convenience, we define $f(\tau):=\sum_{q'=0}^K P_{Kq'}(\tau) \left(K - q'\right) = K-E(\tau)$. Since $\pi^*$'s action is to probe immediately for $\tau \geq \tau^*$,
\begin{equation*}
	V_\gamma(\tau) = K - E(\tau) + \beta V_\gamma(1) =f(\tau) + \beta V_\gamma(1).
\end{equation*}
Similarly, for $\tau <\tau^*$, since $\pi^*$'s action is to idle, we have
\begin{align*}
	V_\gamma(\tau) &= \gamma + \beta V_\gamma(\tau+1) = \gamma + \beta \gamma +\beta^2 V_\gamma(\tau+2) = ... \\
	& = \gamma \frac{1-\beta^{\tau^*-\tau}}{1-\beta} + \beta^{\tau^*-\tau} V_\gamma(\tau^*) \\
&	= \gamma \frac{1-\beta^{\tau^*-\tau}}{1-\beta} + \beta^{\tau^*-\tau} f(\tau^*) +  \beta^{\tau^*-\tau+1} V_\gamma(1).
\end{align*}
Plugging in $\tau=1$, we can solve for $V_\gamma(1)$ and express the value function under the theorem's policy $\pi^*$ as follows
\begin{equation*}
	V_\gamma(\tau) = \begin{cases}
		\frac{\gamma}{1-\beta} - \frac{\beta^{\tau^*-\tau}}{1-\beta^{\tau^*}}\gamma + \frac{\beta^{\tau^*-\tau}}{1-\beta^{\tau^*}}f(\tau^*) & \text{ for } \tau < \tau^*, \\
		f(\tau) + \frac{\beta (1-\beta^{\tau^*-1})}{(1-\beta)(1-\beta^{\tau^*})}\gamma + \frac{\beta^{\tau^*}}{1-\beta^{\tau^*}}f(\tau^*) & \text{ for } \tau \geq \tau^*.
	\end{cases}
\end{equation*}

Next, we will show that $\pi^*$'s value function calculated above indeed satisfies the Bellman equation \eqref{eq:value_func_model1}. Let $\Delta^{\pi}(\tau) := V_\gamma(\tau; \pi(\tau)=1) - V_\gamma(\tau; \pi(\tau)=0)$ be the difference in values of probing and idling at any $\tau$. To verify that $\pi^*$ satisfies Bellman equations, we must show that $\Delta^{\pi^*}(\tau) \geq 0$ for $\tau\geq\tau^*$ and $\Delta^{\pi^*}(\tau) < 0$ for $\tau<\tau^*$. Indeed, evaluating $\Delta^{\pi^*}(\tau)$ for $\tau<\tau^*$, we have
\begin{equation*}
	\Delta^{\pi^*}(\tau) = f(\tau) - f(\tau^*) + \frac{1-\beta^{\tau^*-\tau}}{1-\beta^{\tau^*}} [f(\tau^*) - \gamma]
\end{equation*}

Now, recall that the policy $\pi^*$'s choice of threshold $\tau^*$ satisfies $W(\tau^*) \leq \gamma \leq W(\tau^* + 1).$ Plugging in the definition of $W(\tau)$ and rearranging, we have $f(\tau^*) - \gamma  \leq \frac{1-\beta^{\tau^*}}{1-\beta} [f(\tau^*-1) - f(\tau^*)]$. Using this inequality in $\Delta^{\pi^*}(\tau)$, we get
\begin{align*}
	\Delta^{\pi^*}(\tau) &\leq f(\tau) - f(\tau^*) + \frac{1-\beta^{\tau^*-\tau}}{1-\beta} [f(\tau^*-1) - f(\tau^*)] \\
	& \hspace*{-0.7cm}= E(\tau^*) - E(\tau) + \frac{1-\beta^{\tau^*-\tau}}{1-\beta} [E(\tau^*) - E(\tau^*-1)] \leq 0
\end{align*}
where, the last inequality is due to Lemma \ref{lemma:E_monotone}.

Similarly, evaluating $\Delta^{\pi^*}(\tau)$ for $\tau\geq\tau^*$, we have
\begin{equation*}
	\Delta^{\pi^*}(\tau) = f(\tau)-\beta f(\tau+1) + \frac{1-\beta}{1-\beta^{\tau^*}}[\beta^{\tau^*}f(\tau^*)-\gamma].
\end{equation*}

Again, using the fact that $\pi^*$'s choice of threshold $\tau^*$ satisfies $W(\tau^*) \leq \gamma \leq W(\tau^* + 1)$, we have $f(\tau^*+1) - \gamma \geq \frac{1-\beta^{\tau^*+1}}{1-\beta}[f(\tau^*+1) - f(\tau^*)]$. Using this in $\Delta^{\pi^*}(\tau)$, we have after some rearrangements
\begin{align*}
\Delta^{\pi^*}(\tau) &\geq f(\tau)- \beta f(\tau+1) - f(\tau^*) + \beta f(\tau^*+1) \\
&=  [E(\tau^*) - E(\tau^*+1)] -[E(\tau) - E(\tau+1)] \\
&\hspace*{1.9cm}+ (1-\beta)[E(\tau^*+1) - E(\tau+1)] \geq 0
\end{align*}
where,  the last inequality is due to Lemma \ref{lemma:E_monotone}. This shows that the value function under $\pi^*$ satisfies the Bellman equation hence we conclude that $\pi^*$ is an optimal policy to model 1's single-arm decoupled problem. \hfill\IEEEQEDhere
\fi
\section{Proof of Proposition \ref{prop:L_structure}}\label{proof:L_structure}
Evaluating the drift $\Delta L_{(q,\tau)}(l)$, we have
\begin{align*}
	\Delta& L_{(q,\tau)}(l) := L_{(q,\tau)}(l) - L_{(q,\tau)}(l+1) \\
	&= \beta^l \left( -\gamma + \sum_{q'=0}^K \left( P_{qq'}(\tau+l) - \beta P_{qq'}(\tau+l+1) \right) \wtilde{V}_\gamma(q') \right) \\
	&= \beta^l \sum_{q'=0}^K P_{qq'}(\tau+l) x_{q'}.
\end{align*}
where, $x_{q'} := -\gamma + \wtilde{V}_\gamma(q') - \beta \sum_{q''=0}^{K} P_{q'q''}(1) \wtilde{V}_\gamma(q'').$

For $q'=0$, we have $x_0 = (1-\beta) (K+\beta V(K,1)) - \gamma.$ Now, since the optimal reward is lower bounded by an always-idle policy's reward, we have $V(K,1)\geq \gamma/(1-\beta)$, therefore $x_0\geq0.$ Further, for $q'>0$, we investigate $x_{q'}$ under two cases. Case 1, the optimal action at $(q',1)$ is idle: Here we have $x_{q'} = -(1-\beta)\gamma + \beta [\beta V_\gamma(q',2) - \sum_{q''=0}^K P_{q'q''}(1)\wtilde{V}_\gamma(K)]$ by plugging in idle action into $\wtilde{V}_\gamma(q')$. Now using Lemma \ref{lemma:levin}, we have $x_{q'} < 0$. Case 2, the optimal action at $(q',1)$ is probe: Here we have $x_{q'} = -\gamma$ by plugging in probe action into $\wtilde{V}_\gamma(q')$.
In summary, we have $x_0>0$ and $x_{q'}<0$ for $q'>0$. Hence, due to Lemma \ref{lemma:rosenfield}, we have $\Delta_{(q,\tau)}(l)$ crosses zero only once in $l$ and $q$ which concludes the proof.
\hfill\IEEEQEDhere
\ifisreport
\section{Proof of Theorem \ref{thm:model2_threshold}}\label{proof:model_2_threshold}
Define $l^*(q,\tau) := \arg \max_{l\geq0} L_{(q,\tau)}(l)$ as the optimal number of time-steps to wait starting from state $(q,\tau)$. Here, ties are broken by selecting the smallest index. Notice that the optimal policy is now given by $\pi^*(q,\tau) = \Ind[l^*(q,\tau) = 0]$ where $\Ind[\cdot]$ is the indicator function. The following lemma shows an important property of $l^*(q,\tau)$.
\begin{lemma}\label{lemma:l_q}
	$l^*(q,\tau+1) = [l^*(q,\tau)-1]^+$ for all $\forall \tau=1,2,...$ where  $[\cdot]^+ := \max\{\cdot, 0 \}$.
\end{lemma}
\begin{IEEEproof}
Writing the expression for $L_{(q,\tau+1)}(l)$ and expressing in terms of $L_{(q,\tau)}(l)$, we get $L_{(q,\tau+1)}(l) = [L_{(q,\tau)}(l+1) - \gamma]/\beta$. Therefore, $l^*(q,\tau+1) = \arg \max_{l\geq0} L_{(q,\tau+1)}(l) =  \arg \max_{l\geq0} [L_{(q,\tau)}(l+1) - \gamma]/\beta =  \arg \max_{l\geq0} L_{(q,\tau)}(l+1) =  \arg \max_{l\geq 1} L_{(q,\tau)}(l) - 1$. In summary, we have
\begin{equation} \label{eq:l_base_eq}
	l^*(q,\tau+1) = \arg \max_{l\geq 1} L_{(q,\tau)}(l) - 1
\end{equation}

Case (i) $l^*(q,\tau) > 0$: By definition, $\arg \max_{l\geq0} L_{(q,\tau)}(l) = l^*(q,\tau) > 0$. Therefore, we can write $\arg \max_{l\geq1} L_{(q,\tau)}(l) = \arg \max_{l\geq0} L_{(q,\tau)}(l)$. Plugging this in \eqref{eq:l_base_eq}, we have $l^*(q,\tau+1) = \arg \max_{l\geq 0} L_{(q,\tau)}(l) - 1 = l^*(q,\tau)-1.$

Case (ii-a) $l^*(q,\tau)=0$ and $L_{(q,\tau)}(0) > L_{(q,\tau)}(1)$: Since $L_{(q,\tau)}(0) > L_{(q,\tau)}(1)$, we know from Proposition \ref{prop:L_structure} that $L_{(q,\tau)}(l) > L_{(q,\tau)}(l+1)$ for all $l\geq1$. Therefore, we get $\arg \max_{l\geq1} L_{(q,\tau)}(l) = 1$. Plugging this in \eqref{eq:l_base_eq}, we get $l^*(q,\tau+1) = \arg \max_{l\geq 1} L_{(q,\tau)}(l) - 1 = 0$.

Case (ii-b) $l^*(q,\tau)=0$ and $L_{(q,\tau)}(0) = L_{(q,\tau)}(1)$: Since $\arg \max_{l\geq0} L_{(q,\tau)}(l) = l^*(q,\tau) = 0$ by definition, for all $l=1,2,...$, we have $L_{(q,\tau)}(l) \leq L_{(q,\tau)}(0) = L_{(q,\tau)}(1)$. Therefore, $\arg \max_{l\geq1} L_{(q,\tau)}(l) = 1$. Plugging this in \eqref{eq:l_base_eq}, we get $l^*(q,\tau+1) = \arg \max_{l\geq 1} L_{(q,\tau)}(l) - 1 = 0$.

Combining these cases, we have $l^*(q,\tau+1) = [l^*(q,\tau)-1]^+$. This concludes proof of Lemma \ref{lemma:l_q}.
\end{IEEEproof}

Using Lemma \ref{lemma:l_q}, we can now show that the optimal policy is a threshold policy with thresholds $\tau^*_q := l^*(q,1) + 1$. Indeed, using $l^*(q,\tau+1) = [l^*(q,\tau)-1]^+$ recursively, we have $\pi^*(q,\tau) = \Ind[l^*(q,\tau) = 0] = \Ind[(l^*(q,1)- \tau +1)^+ = 0] = \Ind[(\tau^*_q-\tau)^+=0] = \Ind[\tau \geq \tau^*_q].$ This shows that a threshold policy with thresholds $\tau^*_q := l^*(q,1) + 1$ is optimal.

We now prove the theorem's second part i.e., $\tau_{q}^* \leq \tau_{q'}^*$ for $q \leq q'$. From above, we know $\tau^*_q = l^*(q,1) + 1$. If $l^*(q,1) = 0$, then $\tau^*_q=1$ and $\tau^*_{q'} = l^*(q',1) + 1 \geq 1 = \tau^*_q$ is trivially true. So, assume $l^*(q,1) > 0$. Since $l^*(q,1) = \arg \max_{l\geq0} L_{(q,1)}(l)$ with ties broken by selecting the smallest index, we have $L_{(q,1)}(l^*(q,1) - 1) < L_{(q,1)}(l^*(q,1))$. Hence, due to Proposition \ref{prop:L_structure}'s second part, $\forall q'\geq q$ we have
\begin{equation}\label{eq:threshold_monotonicity_contradiction}
	L_{(q',1)}(l^*(q,1) - 1) < L_{(q',1)}(l^*(q,1)).
\end{equation}

Now assume that $\tau_{q}^*>\tau_{q'}^*$ for some $q'\geq q$. In  other words, $\tau_{q}^*-1 \geq \tau_{q'}^*$. Since $\tau^*_q = l^*(q,1) + 1$, we can now write $l^*(q,1)-1 \geq l^*(q',1) = \arg \max_{l\geq0} L_{(q',1)}(l)$. Hence, from Proposition \ref{prop:L_structure}'s first part, we have  $L_{(q',1)}(l) \geq L_{(q',1)}(l +1)$ for all $l\geq l^*(q,1)-1$. However this contradicts \eqref{eq:threshold_monotonicity_contradiction} hence proving that $\tau_{q}^* \leq \tau_{q'}^*$ for all $q \leq q'$.
 \hfill\IEEEQEDhere
\fi
\section{Proof of Proposition \ref{prop:active_condition}}\label{proof:active_condition}
We prove this by contradiction similar to the proof of the total discounted passive time condition in \cite{spectrum_sharing_main}. Assume contradiction that $\exists \; \gamma_0,q_0$ such that optimal threshold $\tau^*_{q_0}(\gamma)$ is decreasing at $\gamma=\gamma_0$. Since $\tau^*_{q_0}(\gamma)$ is the threshold, we have $V_{\gamma_0}^{\pi_1}(q=q_0, \tau=\tau^*_{q_0}(\gamma_0)) \geq V_{\gamma_0}^{\pi_0}(q=q_0, \tau=\tau^*_{q_0}(\gamma_0))$ where, the superscripts $\pi_1$ and $\pi_0$ indicate the policies used. Since threshold $\tau^*_{q_0}(\gamma)$ is decreasing at $\gamma=\gamma_0$, there exists some $\epsilon>0$ such that $\tau_{q_0}^*(\gamma_0) < \tau_{q_0}^*(\gamma_0-\epsilon)$. Therefore, for a problem with subsidy $\gamma=\gamma_0-\epsilon$, it is optimal to idle when $(q,\tau) = (q_0, \tau_{q_0}(\gamma_0))$. In other words, we have $V_{\gamma_0-\epsilon}^{\pi_1}(q=q_0, \tau=\tau^*_{q_0}(\gamma_0)) < V_{\gamma_0-\epsilon}^{\pi_0}(q=q_0, \tau=\tau^*_{q_0}(\gamma_0))$. Combining the two value function inequalities obtained above, we have $\partial V_\gamma^{\pi_1}(q_0,\tau_{q_0}^*)/\partial\gamma > \partial V_\gamma^{\pi_0}(q_0,\tau_{q_0}^*)/\partial\gamma$ at $\gamma=\gamma_0$. Using the fact that the partial derivative of optimal value w.r.t. subsidy is equal to the total discounted passive time, this inequality can be rewritten as $A_{\pi_1}(q_0,\tau_{q_0}^*) < A_{\pi_0}(q_0,\tau_{q_0}^*)$.
However, this contradicts the Proposition's statement.
\hfill\IEEEQEDhere
\section{Proof of Theorem \ref{thm:indexability_smallK}}\label{proof:indexability_smallK}

The blocking queues model and the standard queues models are equivalent when $K=1$. Since we have already shown indexability for the standard queue model, we focus on $K=2$. Let $\wtilde{A}_q$ be the active time under the optimal policy given that an observation of $q$ was just made. After observing a backlog $q$, the optimal policy waits $\tau_q^*$ time-steps and probes again, hence we have $\forall q\in\{1...,K\}$,
$
	\wtilde{A}_q = 1 + \beta^{\tau^*_q} \sum_{q'=0}^K P_{qq'}(\tau_q^*)\wtilde{A}_{q'} +  \beta^{\tau^*_q} P_{q0}(\tau_q^*)\wtilde{A}_{K}
$
where, we used $\wtilde{A}_0 = \wtilde{A}_K$ since after seeing zero backlog, the dispatcher fills the queue to $K$ and from then the system behaves exactly the same as observing $K$. Setting $K=2$ in the above equations and rearranging, we have
\begin{equation}\label{eq:active_time_basis_1}
	\wtilde{A}_1-\wtilde{A}_2 = \frac{(1-(1-\beta^{\tau^*_1})) \wtilde{A}_2}{1-\beta^{\tau_1^*} P_{11}(\tau_1^*)} = \frac{-(1-(1-\beta^{\tau^*_2})\wtilde{A}_2)}{\beta^{\tau_2^*}P_{21}(\tau^*_2)}.
\end{equation}
Further, under the optimal policy, the waiting time between any two probes is one of $\{\tau^*_1,..., \tau^*_K\}$. Further, since $\tau^*_1\leq \tau^*_2 ...\leq \tau^*_K$, the sparsest possible activation pattern is all probes being separated by $\tau^*_K$ time-steps. Similarly, the densest possible activation pattern is all probes being separated by $\tau^*_1$ time-steps. Therefore, $\forall q$ we have
\begin{equation}\label{eq:active_time_basis_2}
	\frac{1}{1-\beta^{\tau_K^*}} \leq \wtilde{A}_q \leq \frac{1}{1-\beta^{\tau_1^*}}.
\end{equation}

Now, since $\pi_1$ probes immediately, we have $A_{\pi_1}(q,\tau_q^*) = \sum_{q'=0}^K P_{qq'}(\tau^*_{q})A_{q'}$. And since $\pi_0$ idles initially, in the next time-step its state is $(q,\tau^*_q+1)$ and since it acts optimally at this time-step, it must probe and hence $A_{\pi_0}(q,\tau_q^*) = \beta\sum_{q'=0}^K P_{qq'}(\tau^*_{q}+1)A_{q'}$. Plugging in $K=2$ and simplifying,
\begin{multline*}
	\Delta A(q,\tau^*_q) := A_{\pi_1}(q,\tau_q^*) - A_{\pi_0}(q,\tau_q^*) \\
	= (1-\beta)\wtilde{A}_2 + (P_{q1}(\tau_q^*) - \beta P_{q1}(\tau_q^*+1)) (\wtilde{A}_1-\wtilde{A}_2).
\end{multline*}

For $q=1$, since departures are Poisson, $P_{11}(\tau)=e^{-\mu\tau}$. Therefore, $\Delta A(1,\tau^*_1) = (1-\beta)\wtilde{A}_2 + e^{-\mu\tau^*_1}(1-\beta e^{-\mu})(\wtilde{A}_1-\wtilde{A}_2).$ Using \eqref{eq:active_time_basis_1} and \eqref{eq:active_time_basis_2}, it can be seen that  $\Delta A(1,\tau^*_1)\geq 0$. For $q=2$, using Poisson departures, $P_{21}(\tau) = \mu\tau e^{-\mu\tau}$. Plugging this in and using \eqref{eq:active_time_basis_1}, $\Delta A(2,\tau^*_2) = (1-\beta)\wtilde{A}_2 - \beta^{-\tau_2^*}(1-(1+1/\tau_2^*)\beta e^{-\mu})(1-(1-\beta^{\tau^*_2})\wtilde{A}_2).$ Now, using \eqref{eq:active_time_basis_2}, $e^{-\mu}\leq1$, and $\beta^{\tau_2^*}\geq 1-\tau_2^*(\beta^{-1}-1)$, we can show that $\Delta A(2,\tau^*_2) \geq (1-\beta)/(1-\beta^{\tau_2^*}) \geq 0.$ Combining these, we have $A_{\pi_1}(q,\tau_q^*) \geq A_{\pi_0}(q,\tau_q^*)$ for both $q=1,2$. Hence, we have indexability due to Proposition \ref{prop:active_condition}.
\hfill\IEEEQEDhere

\bibliographystyle{IEEEtran}
\bibliography{content/refs}

@misc{weber_jsq,
	note = {R. R. Weber, ``On the optimal assignment of customers to parallel servers,'' \emph{Journal of Applied Probability}, vol. 15, no. 2, pp. 406--413, 1978.}
}

@misc{mitz_jsq_d,
	note = {M. Mitzenmacher, ``The power of two choices in randomized load balancing,'' \emph{IEEE Transactions on Parallel and Distributed Systems}, vol. 12, no. 10, pp. 1094--1104, Oct. 2001, doi: 10.1109/71.963420.}
}

@misc{jiq_1,
	note = {R. Badonnel and M. Burgess, ``Dynamic pull-based load balancing for autonomic servers,'' in \emph{NOMS 2008 IEEE Network Operations and Management Symposium}, 2008, pp. 751--754.}
}

@article{belief_sufficient_stat,
	title={The optimal control of partially observable Markov processes over a finite horizon},
	author={Smallwood, Richard D and Sondik, Edward J},
	journal={Operations research},
	volume={21},
	number={5},
	pages={1071--1088},
	year={1973},
	publisher={INFORMS}
}

@article{whittle,
	title={Restless bandits: Activity allocation in a changing world},
	author={Whittle, Peter},
	journal={Journal of applied probability},
	volume={25},
	number={A},
	pages={287--298},
	year={1988},
	publisher={Cambridge University Press}
}

@article{whittle_optimality,
	title={On an index policy for restless bandits},
	author={Weber, Richard R and Weiss, Gideon},
	journal={Journal of applied probability},
	volume={27},
	number={3},
	pages={637--648},
	year={1990},
	publisher={Cambridge University Press}
}

@article{whittle_condition_nino,
	title={Restless bandits, partial conservation laws and indexability},
	author={Nino-Mora, Jose},
	journal={Advances in Applied Probability},
	volume={33},
	number={1},
	pages={76--98},
	year={2001},
	publisher={Cambridge University Press}
}

@article{whittle_condition_akbarzadeh,
	title={Conditions for indexability of restless bandits and an algorithm to compute Whittle index},
	author={Akbarzadeh, Nima and Mahajan, Aditya},
	journal={Advances in Applied Probability},
	volume={54},
	number={4},
	pages={1164--1192},
	year={2022},
	publisher={Cambridge University Press}
}

@misc{known_queues_glazebrook,
	note = {K. D. Glazebrook, C. Kirkbride, and J. Ouenniche, ``Index policies for the admission control and routing of impatient customers to heterogeneous service stations,'' \emph{Operations Research}, vol. 57, no. 4, pp. 975--989, 2009.}
}

@article{akbarzadeh_machine,
	title={Maintenance of a collection of machines under partial observability: Indexability and computation of whittle index},
	author={Akbarzadeh, Nima and Mahajan, Aditya},
	journal={Les Cahiers du GERAD ISSN},
	volume={711},
	pages={2440},
	year={2021}
}

@article{vishrant_age,
	title={A whittle index approach to minimizing functions of age of information},
	author={Tripathi, Vishrant and Modiano, Eytan},
	journal={IEEE/ACM Transactions on Networking},
	volume={32},
	number={6},
	pages={5144--5158},
	year={2024},
	publisher={IEEE}
}

@inproceedings{xiaojun_age,
	title={An easier-to-verify sufficient condition for whittle indexability and application to AoI minimization},
	author={Zhou, Sixiang and Lin, Xiaojun},
	booktitle={IEEE INFOCOM 2024-IEEE Conference on Computer Communications},
	pages={1741--1750},
	year={2024},
	organization={IEEE}
}

@article{spectrum_sharing_main,
	title={Indexability of restless bandit problems and optimality of whittle index for dynamic multichannel access},
	author={Liu, Keqin and Zhao, Qing},
	journal={IEEE Transactions on Information Theory},
	volume={56},
	number={11},
	pages={5547--5567},
	year={2010},
	publisher={IEEE}
}

@inproceedings{spectrum_sharing_2,
	title={Indexability and whittle index for restless bandit problems involving reset processes},
	author={Liu, Keqin and Weber, Richard and Zhao, Qing},
	booktitle={2011 50th IEEE Conference on Decision and Control and European Control Conference},
	pages={7690--7696},
	year={2011},
	organization={IEEE}
}

@article{rosenfield_machine,
	title={Markovian deterioration with uncertain information},
	author={Rosenfield, Donald},
	journal={Operations Research},
	volume={24},
	number={1},
	pages={141--155},
	year={1976},
	publisher={INFORMS}
}

@article{rosenfield_machine_2,
	title={Markovian deterioration with uncertain information—a more general model},
	author={Rosenfield, Donald},
	journal={Naval Research Logistics Quarterly},
	volume={23},
	number={3},
	pages={389--405},
	year={1976},
	publisher={Wiley Online Library}
}

@techreport{levin_machine,
	title={Optimal Inspection Policies for Deteriorating Markov Processes.},
	author={Levin, Robert D},
	year={1977}
}

\end{document}